\documentclass[twoside,leqno,twocolumn]{article} 
\usepackage[letterpaper]{geometry}

\usepackage{ltexpprt}
\usepackage{natbib}
\usepackage{booktabs} 
\usepackage[utf8]{inputenc} 

\usepackage{graphicx} 

\usepackage{booktabs} 
\usepackage{array} 
\usepackage{paralist} 
\usepackage{verbatim} 
\usepackage{subfig} 

\usepackage{amsmath}
\usepackage{amssymb}
\usepackage{url}
\usepackage{algorithm}
\usepackage{algorithmic}
\usepackage{color}
\newtheorem{lm}{Lemma}

\newtheorem{cor}[lm]{Corollary}

\newtheorem{definition}[lm]{Definition}

\DeclareMathOperator{\argmax}{argmax}
\DeclareMathOperator{\NSW}{NSW}

\newcommand{\utilapprox}{\rho}
\newcommand{\wcutilapprox}{\utilapprox^*}
\newcommand{\mech}{p}
\newcommand{\out}{o}
\newcommand{\val}{v}
\newcommand{\alloc}{p}
\newcommand{\pa}{q}

\newcommand{\bench}[1][*]{p^{#1}}
\newcommand{\nsw}{\bench[\NSW]}

\newcommand{\ceil}[1]{\lceil#1\rceil}
\newcommand{\floor}[1]{\lfloor#1\rfloor}

\newcommand{\UpperBoundShort}{$O(2^{2\sqrt{\log n}})$}
\newcommand{\UpperBound}{$O(2^{2\sqrt{\log n}})\subseteq o(n^{\epsilon})$ for any constant $\epsilon>0$}

\newcolumntype{C}[1]{>{\centering\arraybackslash}m{#1}}

\usepackage{tikz}

\newcommand{\valFar}{\val^{a,F}_r}

\newcommand{\valFcr}{\val^{c,F}_r}

\newcommand{\pAFa}{t^{A,F}_1}

\newcommand{\vij}{v_{i,j}}
\newcommand{\pij}{p_{i,j}}

\newcommand{\kr}{k_r}

\newcommand{\Ms}{M_s}

\newcommand{\sar}{s_{a,r}}
\newcommand{\sbr}{s_{b,r}}
\newcommand{\scr}{s_{c,r}}
\newcommand{\sdr}{s_{d,r}}

\newcommand{\sfr}{s_{f,r}}
\newcommand{\sgr}{s_{g,r}}
\newcommand{\shr}{s_{h,r}}

\newcommand{\vIfr}{v^I_{f,r}}
\newcommand{\vIgr}{v^I_{g,r}}
\newcommand{\vIhr}{v^I_{h,r}}

\newcommand{\vFfr}{v^F_{f,r}}
\newcommand{\vFgr}{v^F_{g,r}}
\newcommand{\vFhr}{v^F_{h,r}}

\newcommand{\vr}{v_r}

\newcommand{\hide}[1]{}

\begin{document}
\title{\Large A Truthful Cardinal Mechanism for One-Sided Matching\thanks{Rediet Abebe was supported in part by a Facebook scholarship. Richard Cole was supported in part by NSF grants CCF-1527568 and CCF-1909538. Vasilis Gkatzelis was supported in part by NSF grant CCF-1755955. Jason Hartline was supported by NSF grant CCF-1618502 and part of this work took place while he was visiting Harvard University.}}
\author{Rediet Abebe \thanks{Society of Fellows, Harvard University}
\and Richard Cole \thanks{Courant Institute, New York University}
\and Vasilis Gkatzelis \thanks{Computer Science Department, Drexel University}
\and Jason D. Hartline \thanks{Computer Science Department, Northwestern University}}
\date{}                     



\maketitle

\fancyfoot[R]{\scriptsize{Copyright \textcopyright\ 2020 by SIAM\\
Unauthorized reproduction of this article is prohibited}}

\begin{abstract} \small\baselineskip=9pt
We revisit the well-studied problem of designing mechanisms for one-sided matching markets, 
where a set of $n$ agents needs to be matched to a set of $n$ heterogeneous items. 
Each agent $i$ has a value $v_{i,j}$ for each item $j$, and these values are private 
information that the agents may misreport if doing so leads to a preferred outcome. 
Ensuring that the agents have no incentive to misreport requires a careful design of 
the matching mechanism, and mechanisms proposed in the literature mitigate this issue by
eliciting only the \emph{ordinal} preferences of the agents, i.e., their ranking of 
the items from most to least preferred. However, the efficiency guarantees of these mechanisms
are based only on weak measures that are oblivious to the underlying values. In this paper we 
achieve stronger performance guarantees by introducing a mechanism that truthfully elicits the 
full \emph{cardinal} preferences of the agents, i.e., all of the $v_{i,j}$ values. We evaluate 
the performance of this mechanism using the much more demanding Nash bargaining solution as a benchmark,
and we prove that our mechanism significantly outperforms all ordinal mechanisms (even non-truthful ones). 
To prove our approximation bounds, we also study the population monotonicity of the Nash 
bargaining solution in the context of matching markets, providing both upper and lower bounds 
which are of independent interest.

%
\end{abstract}




\section{Introduction}



%
%
In this paper we consider the classic ``house allocation'' problem of
\citet{HZ-79}. A set of $n$ agents are to be matched, one-to-one, to a
set of $n$ items and each agent $i$ has a value $v_{i,j}$ for each item $j$.
A randomized matching mechanism outputs a probability distribution over matchings, 
which corresponds to a doubly-stochastic matrix $p$, providing the
probability $p_{i,j}$ that $i$ will be matched to $j$; the expected
utility of $i$ in $p$ is $\sum_j v_{i,j}p_{i,j}$. The goal in this setting
is to generate a fair and efficient randomized matching, which crucially
depends on the values of the agents, and the main obstacle is the fact that 
the $v_{i,j}$ values of each agent $i$ are private information of this agent.
Therefore, a successful mechanism needs to elicit the agents' preferences and 
output a desired matching; on the other hand, each agent's goal is to 
maximize her expected utility, so an agent can strategically misreport her
preferences if this increases her expected utility. 

This tension between the objectives of the designer and those of the participants
lies at the core of the sub-field of economics known as mechanism design. Most of 
the proposed solutions in the mechanism design literature, however, leverage monetary 
payments as the main tool that the designer can use to incentivize truthful reporting 
by the agents. This is in contrast to the matching mechanisms in the house
allocation setting which cannot use such payments, making the mechanism design problem 
considered in this paper particularly demanding. In the absence of monetary payments,
an alternative tool for simulating the impact of payments is to use ``money burning''~\citep{hr}. 
In our setting, this could correspond to keeping some of the items unmatched with positive
probability, thus penalizing the agents just like monetary payments would. But, in many
settings, including the house allocation problem studied in this paper, this would be 
unacceptable (e.g., it would imply that some agents may remain homeless while some houses 
remain unoccupied).

Our main result is a novel application of random sampling that enables the use of known 
money burning mechanisms, while ensuring that every agent and item is matched. In other words, 
our technique takes advantage of the improved incentives that these money-burning mechanisms provide,
but without suffering their most important drawback. To verify the usefulness of this approach,
we combine it with the {\em partial allocation} (PA) mechanism of \citet{CGG-13}, giving rise 
to a new mechanism that incentivizes the agents to always truthfully report their \emph{cardinal} 
preferences, i.e., their $v_{i,j}$ values, and yields an outcome $p$ that is approximately both 
fair and efficient.\footnote{\citet{zho-90} shows that there is no mechanism
  that is truthful, symmetric, and Pareto efficient; thus, some notion
  of approximation is necessary.}  We measure the performance of our
mechanism using the canonical benchmark defined by the Nash bargaining
solution and show that our mechanism outperforms the standard mechanisms 
with the same, or weaker, incentive properties.

%
%
The literature on one-sided matching has considered three main approaches, none of
which gives rise to mechanisms that are both truthful and obtain a non-trivial approximation 
of the aforementioned benchmark. \citet{HZ-79} propose the {\em competitive equilibrium from equal
incomes} (CEEI), which depends on the $v_{i,j}$ values in a non-trivial way, but it provides the agents 
with strong incentives to misreport these values, especially for small problem instances.
The {\em random serial dictatorship} (RSD), 
or random priority, mechanism is an important mechanism with a long history in practice. 
This mechanism randomly orders the agents and, following this order, gives to each agent
her favorite item among the ones that are still available. RSD is an \emph{ordinal} mechanism:
it requires only the ordinal preferences of each agent, i.e., only her ranking of the items 
from most to least preferred. It elicits this information truthfully, but its outcomes can 
be very inefficient. The {\em probabilistic serial} (PS) mechanism of \citet{BM-01} is
another ordinal mechanism, and its outcome is computed by continuously allocating to each 
agent portions of her most preferred item that has not already been fully allocated. PS satisfies 
an ordinal notion of efficiency, but it achieves only a trivial approximation of our stronger 
benchmark, and it is not truthful. 
We provide a more detailed discussion regarding these mechanisms and other related work
in Section~\ref{sec:related}.

%
%
Aiming to provide stronger efficiency and fairness guarantees
compared to known mechanisms, we consider a cardinal benchmark: 
the well-studied \emph{Nash bargaining solution},
proposed by~\citet{Nash50}. Given a \emph{disagreement point}, i.e.,
the ``status quo'' that would arise if negotiations among the agents
were to break down, the Nash bargaining solution is the outcome that
maximizes the product of the agents' marginal utilities relative to
their utility for the disagreement point. This outcome indicates the
utility that each agent ``deserves'', so we use this utility as
the benchmark for that agent.  The choice of disagreement point can
depend on the application at hand: if a buyer and a seller are
negotiating a transaction, the disagreement point could be that the
seller keeps the goods and the buyer keeps her money. In one-sided
matching markets the disagreement point needs to be a matching because
leaving an agent without a house is infeasible. Since all agents have
symmetric claims on the items when entering the market, we let the
disagreement point be a matching chosen uniformly at random,
ensuring that each agent is equally likely to be matched to each
item. The Nash bargaining solution therefore corresponds to the
doubly-stochastic matrix $p$ that maximizes $\prod_i (\sum_j
v_{i,j}p_{i,j} -o_i)$, where $o_i=\frac{1}{n}\sum_{j}v_{i,j}$ is the
expected utility of agent $i$ for an item chosen uniformly at
random. 

Since no truthful and symmetric mechanism can guarantee Pareto efficiency
\citep{zho-90}, it is clearly impossible for a truthful mechanism to
implement the Nash bargaining solution, which is symmetric and
Pareto efficient. Thus, we consider the problem of approximating this
solution. Specifically, the Nash bargaining solution defines the utility
that each agent deserves and our goal is to ensure that every agent
receives a good approximation of that benchmark. Formally, a mechanism is a
$\beta$-approximation if the utility of each agent is at least a
$\beta$ fraction of her utility in the Nash bargaining solution.
Note that, once the valuations of each agent $i$ are adjusted
by subtracting $o_i$, then our objective corresponds to the {\em Nash
  social welfare} (NSW), which has recently received a lot of
attention in the fair division literature
(e.g.,~\citealp{CG-18,GHM18,CKMPSW-16,BKV18,BGM-17}). The NSW maximizing
outcome is {\em proportionally fair} in that it satisfies a
multiplicative version of Pareto efficiency, namely, the utility of an
agent cannot be increased by a multiplicative factor without
decreasing the product of utilities of other agents by a greater
multiplicative factor.


En route to proving our mechanism's approximation bounds, we also provide
an analysis of the Nash bargaining solution with respect to its {\em population
monotonicity}, which is of independent interest. It has long been known that, 
unlike the Kalai-Smorodinsky solution, the Nash bargaining solution can violate 
population monotonicity for some instances of the bargaining problem~\citep{T-83,TL-89}. 
That is, there exist instances where
removing some of the agents and computing the updated Nash bargaining
solution can decrease the utility of some of the remaining agents.
When allocating items among competing agents, this lack of
monotonicity is somewhat counterintuitive.  Why would the decreased
competition from agents departing the market not lead to (weakly)
increased utility for the agents remaining in the market?  Indeed, we show that
population monotonicity can be violated in the Nash bargaining
solution for matching markets.  Effectively, the constraint that the
allocation is a distribution over perfect matching introduces
positive externalities between agents.

In order to quantify the extent to which one of the remaining agents'
utility can drop after such a change in the agent population, the
bargaining literature in economics introduced the \emph{opportunity
  structure} notion~(e.g., see the book by \citealp{TL-89}, and
references therein). This structure identifies the largest factor by
which a remaining agent's utility can drop after some subset of agents
is removed. In fact, resembling the standard computer science
approach, the opportunity structure is defined as the
\emph{worst-case} factor over all instances, all removed subsets of
agents, and all remaining agents. In this paper we provide essentially
tight upper and lower bounds for this factor in the context of matching markets, 
showing that in carefully designed worst-case instances, this factor can
grow faster than a polylogarithmic function of the number of agents, yet
slower than any polynomial. Apart from the broader interest in
understanding this measure in matching markets, we show that the upper
bound on the population non-monotonicity provides, up to constant
factors, an upper bound on the approximation factor of the truthful
matching mechanism that we define.


\subsection*{Our Results.}
%
%
In this paper we introduce a random sampling technique which
allows us to translate non-trivial truthful one-sided matching
mechanisms that may produce \emph{partial} matchings (i.e., possibly
leaving some agents unmatched) into ones where (i) every agent is
always assigned an item, and (ii) the incentives for truthful
reporting of preferences are maintained. For example, the truthfulness 
guarantee of the PA mechanism of \citet{CGG-13} depends heavily on its ability to
penalize the agents that cause inconvenience to others; it thereby
ensures that none of these agents is misreporting their
preferences. Since monetary payments are prohibited, this mechanism
penalizes the agents by assigning positive probability to outcomes 
that leave them unmatched. Such partial matchings, however, are unacceptable 
in the house allocation problem. Every agent,
no matter what values she reports, needs to be guaranteed an item, and
this constraint significantly restricts our ability to introduce
penalties. Nevertheless, we show that we can still recreate such
penalties by using random sampling. Applying our sampling technique to
the PA mechanism, we define the \emph{randomized partial improvement} (RPI)
mechanism, which significantly outperforms all the standard matching
mechanisms with respect to the Nash bargaining benchmark.

In essence, the RPI mechanism endows agents with a baseline allocation given
by a uniformly random item and then uses the PA mechanism to improve the agents' 
utility relative to this baseline. In reality, it is not possible to simultaneously 
maintain the baseline and offer improvements to all agents, so RPI circumvents this 
impossibility by imposing these two conditions on a sample of half of the agents instead.  
With half the agents (but all of the items) there is sufficient flexibility to
faithfully implement the PA mechanism with the outside option of a uniform random house.  
After finalizing the allocation of the sampled agents, RPI then recursively allocates 
the unallocated portions of the items to the remaining agents.

%
%

As an intermediate step toward the theoretical analysis of RPI's approximation
factor, we study the extent to which population monotonicity may be violated in 
a one-sided matching market instance. We refer to an instance as $\utilapprox$-utility monotonic 
if removing a subset of its agents can \emph{decrease} a remaining agent's utility in 
the new Nash bargaining solution by a factor no more than $\rho$. We show that, for a very carefully 
constructed (and somewhat contrived) family of instances, $\utilapprox$ can be as high as $\Omega(2^{\sqrt{\log n}/2})$
and we complement this bound with an essentially tight upper bound, by proving that for any one-side matching instance $\rho$ is no more
than \UpperBound. 

Apart from the broader interest in understanding the extent to which the Nash bargaining solution may
violate population monotonicity, our upper bound on $\utilapprox$ also directly implies an upper bound for the
approximation factor of RPI. Specifically, we prove that RPI guarantees to {\em every} agent a 
$4\,e\,\utilapprox$ approximation of the utility that she gets in the Nash bargaining benchmark.
Therefore, as a corollary, we conclude that RPI approximates the Nash bargaining benchmark within
\UpperBound, even with the worst case choice of $\rho$.
In stark contrast to this upper bound, which is strictly better than any polynomial, we show that 
the approximation factor of all ordinal mechanisms (even ones that are not truthful, such as probabilistic
serial) grows linearly with the number of agents. Therefore, our mechanism significantly outperforms
all ordinal mechanisms while at the same time satisfying truthfulness. 


\paragraph{Structure.} Section~\ref{s:prelim} provides some preliminary
definitions and Section~\ref{s:benchmark} formally introduces the benchmark
and approximation measure used throughout the paper. Our results showing that
ordinal mechanisms fail to achieve any non-trivial approximation are in Section~\ref{s:ordinal};
Section~\ref{s:mechanism} includes the description of our mechanism and the proofs
regarding its truthfulness and fairness guarantees. Finally, in Section~\ref{s:monotonicity}
we study the population monotonicity of the Nash bargaining solution and provide both
upper and lower bounds for it.

%
%
%


\section{Preliminaries}
\label{s:prelim}

Given a set $N$ of $n$ agents and a set $M$ of $n$ items, a randomized
matching can be represented by a doubly-stochastic matrix $\alloc$ of
marginal probabilities, where $\alloc_{i,j}$ denotes the marginal
probability that agent $i$ is allocated item $j$. Clearly, any probability
distribution over matchings implies a double-stochastic matrix, and the 
Birkhoff-von-Neumann theorem shows that \emph{any} doubly-stochastic matrix 
can be implemented as a probability distribution over matchings. 
Denote by $v$ a matrix of agent values where $v_{i,j}$ is
the value of agent $i$ for item $j$.  The expected utility of agent
$i$ for random matching $p$ is $u_i = \sum_{j\in M} \val_{i,j}\,\alloc_{i,j}$.  
The random matching $p$ that a mechanism outputs when 
the agents' reported values are $\val$ is denoted by $\alloc(\val)$.

For each agent $i$, her values $v_i =(\val_{i,1},\ldots,\val_{i,n})$ are private and a 
matching mechanism must be designed to properly elicit them. A mechanism is {\em truthful} if 
it is a dominant strategy for each agent $i$ to report her true values. That is, if we let 
$\alloc(w_i,\val_{-i})$ denote the outcome of the mechanism when agent $i$ reports values 
$w_i$ and all the other agents report values $\val_{-i}$, then a mechanism is truthful if
for every agent $i$, any matrix of values $\val$, and any misreports $w_i$: 
\[\sum_{j\in M}
\val_{i,j}\,\alloc_{i,j}(\val) ~\geq~ \sum_{j\in M} \val_{i,j}\,\alloc_{i,j}(w_i,\val_{-i}).\]

Our benchmark, formally defined in the following section, uses the {\em Nash social welfare} (NSW)
objective on appropriately adjusted agent valuations. The NSW maximizing outcome is known to 
provide a balance between fairness and efficiency by maximizing the geometric mean (or, 
equivalently, the product) of the agents' expected utilities, i.e., $\max_{\alloc} \prod_i
\left(\sum_j\val_{i,j}\,\alloc_{i,j}\right)$. The partial allocation mechanism from \citet{CGG-13}
provides a truthful approximation of that outcome and can be easily adapted to randomized matchings 
by interpreting fractional allocations as probabilities.

\begin{definition}
The {\em partial allocation} (PA) mechanism on values $\val$ works as follows:
\begin{enumerate}
\item Compute the doubly-stochastic matrix $\nsw(\val)$ that maximizes
  the Nash social welfare.
\item For each agent $i$, compute $f_i$ as follows:
\begin{enumerate}
\item Let $u_k$ be agent $k$'s utility in $\nsw(\val)$.
\item Let $u'_k$ be  agent $k$'s utility in $\nsw(\val_{-i})$, i.e., in the NSW maximizing allocation with 
agent $i$ absent and all other agents restricted to one unit, i.e., $\sum_j \nsw_{k,j}(\val_{-i})=1$ for all $k\neq i$.
\item 
Let $f_i = { \prod\nolimits_{k\ne i} u_k} \,\Big\slash\, { \prod\nolimits_{k\ne i} u'_k}.$
\end{enumerate}
\item Allocate each item $j$ to each agent $i$ with probability $\pa_{i,j} = f_i\,\nsw_{i,j}(\val)$.
\end{enumerate}
\end{definition}
 
Note that the fraction $f_i$ of the NSW maximizing assignment
allocated to agent $i$ is equal to the relative loss in utility that
$i$'s presence imposes on the other agents.  The denominator
is independent of $i$'s declared valuations, so, in maximizing $f_i
\cdot u_i$, which would be agent $i$'s goal, she is maximizing the NSW
when she reports truthfully.  \citet{CGG-13} show that $f_i\in (1/e,
1]$, without the unit constraint on allocations, but the same argument
  holds with the unit constraint.

\begin{theorem}[\citealp{CGG-13}]\label{thm:CGG13}
The partial allocation mechanism is truthful, feasible, and allocates each agent $i$ at fraction $f_i$ of the NSW maximizing assignment, where $f_i$ is at least $1/e$.
\end{theorem}

\section{The Nash Bargaining Benchmark}
\label{s:benchmark}

In this section, we define our cardinal benchmark as well as an
approximation measure for evaluating mechanisms for the one-sided
matching problem. Our benchmark is the Nash bargaining solution with a
uniformly random matching as the disagreement point.  Each agent $i$'s
expected utility for this disagreement point is $\out_{i} =
\frac{1}{n}\sum_{j} \val_{i,j}$ and the Nash bargaining solution is
the outcome $p^*$ that maximizes the Nash Social Welfare
objective with respect to the marginal valuations $\val - o$.  In
other words, the Nash bargaining solution distributes the additional
value, beyond each agent's outside option, in a fair and efficient
manner.

\begin{definition} 
The Nash bargaining solution with disagreement point $(o_i)_{i\in N}$ is
$$
p^* = \argmax_{\alloc} \left\{\prod\nolimits_i \Big(\sum\nolimits_j  \val_{i,j} \,\alloc_{i,j} - \out_i \Big)\right\},
$$ 
where every agent $i$ is constrained to have non-negative utility
$\sum\nolimits_j \val_{i,j} \alloc_{i,j} - \out_i \geq 0$.
\end{definition}

Apart from its fairness properties, this benchmark is also appealing
because of its invariance to additive shifts and multiplicative
scalings of any agent's values for the items. Shifting all the values
of an agent by adding some constant does not affect the marginal
values after the outside option is subtracted.  Also, scaling all of
the values of an agent by some constant does not have any impact on
what the Nash bargaining solution, $p^*$, is; the product value of
every outcome is multiplied by the same constant, and hence the
optimum is unaffected.  As a result, we do not need to assume that the
values reported by the agents are scaled in any particular way.  One
thing to note about the benchmark being invariant to these changes is
that, on instances where the agents' values are identical up to shifts
and scales, the benchmark assignment is the uniform random
assignment.%
\footnote{The combined property of shift and scale invariance has some
  counterintuitive implications.  Consider an example instance where
  all agents $i$ have value $v_{i,1} > 1$ for item 1, and $v_{i,j} =
  1$ for all other items $j \in \{2,\ldots,n\}$.  In the Nash
  bargaining solution, all agents receive a uniform random item and in
  particular a $1/n$ fraction of the preferred item~1.  This outcome
  may seem surprising as it does not account for the possibility that
  some agents may prefer item 1 much more than other agents.  This
  uniform outcome results because the agents' preferences are
  equivalent up to additive and multiplicative shifts.}

Our goal is to approximate $p^*$, the Nash bargaining solution with
disagreement point $(o_i)_{i\in N}$, using the following per-agent 
approximation measure.  

\begin{definition}
The per-agent approximation of mechanism $\mech$ with respect to benchmark
assignment $\bench$ is the worst-case ratio of the utility of any
agent in $\bench$ and $\mech$,
 $$\max_{\val}\left\{ \max_{i} \left\{\frac{\sum_j \val_{i,j} \bench_{i,j}(\val)}{\sum_j \val_{i,j} \mech_{i,j}(\val)}\right\}\right\}.$$
\end{definition}


\section{Inapproximability by Ordinal Mechanisms}
\label{s:ordinal}

Ordinal mechanisms are popular in the literature on matching.  Rather
than asking agents for cardinal values for each item, an ordinal
mechanism need only solicit an agent's preference order over the
items.  Two prevalent ordinal mechanisms are the random serial
dictatorship (RSD) and probabilistic serial (PS) mechanisms.  One
of our main motivations for studying cardinal mechanisms in this paper 
is that ordinal mechanisms are bound to generate unfair allocations for
some instances, due to the fact that they disregard the intensity of the 
agents' preferences; even when the agents agree, or partially agree, on 
their preference order, they may still disagree on preference intensities.
A mechanism that does not take these intensities into consideration is, for
example, unable to distinguish between agents whose favorite item is very 
strongly preferred over the rest, and agents who have only a slight preference 
for their top item over the rest.

Our first lower bound shows that the random serial dictatorship mechanism
can be very unfair to some agent, leading to an approximation factor as
bad as $n$ (the number of agents).

\begin{lemma}
The worst case approximation ratio of the random serial dictatorship
(RSD) mechanism with respect to the Nash bargaining benchmark is $n$.
\end{lemma}

\begin{proof}
Consider the example where agent 1 has value 1 for item 1 and no value for any other item, and each agent 
$i\geq 2$ has value 1 for item 1, value $1-\epsilon$ for item $i$, and no value for other items:
\begin{align*}
\val &= 
\begin{bmatrix}
1      & 0      & 0      & \ldots & 0 \\
1      & 1-\epsilon & 0      & \ldots & 0 \\
1      & 0      & 1-\epsilon & \ldots & 0 \\
\vdots & \vdots & \vdots & \ddots & \vdots \\
1      & 0      & 0      & \ldots & 1-\epsilon 
\end{bmatrix}.
\end{align*}
In RSD an ordering of the agents is generated uniformly at random,
and then each agent is allocated her favorite available item in that
order. In this instance, the first agent in the random ordering will
always select item 1, and every agent has the same probability, $1/n$,
of being ordered first. Since agent 1 has no value for any other item,
the expected utility of this agent 1 in RSD is $1/n$.

On the other hand, as $\epsilon$ approaches zero, the Nash bargaining solution 
assigns each agent $i$ to item $i$ with probability that approaches 1. To verify 
this fact, note that for $\epsilon = 0$
the Nash bargaining would assign agent $i$ to item $i$ with
probability 1, and observe that the distribution that RSD outputs is
continuous in $\epsilon$. Thus, the utility of each agent in the Nash
bargaining solution -- and specifically of agent 1 -- approaches 1.
As a result, the RSD mechanism is being unfair to agent 1, leading to
an approximation factor of $n$.
\end{proof}

In fact, with a small modification of the instance used to verify how
unfair the RSD mechanism can be, the following theorem shows that
\emph{every} ordinal mechanism is susceptible to this issue.

\begin{theorem}
The worst case approximation ratio of any ordinal mechanism
to the Nash bargaining benchmark is at least $n-1$.
\end{theorem}

\begin{proof}
Consider the following instance $\val$, where agents correspond to
rows and items to columns:
\begin{align*}
\val &= 
\begin{bmatrix}
1 &\epsilon& 0 & 0 & \ldots & 0 & 0 \\
1 & 0 & 1-\epsilon& 0 & \ldots & 0 & 0\\
1 & 0 & 0 &1-\epsilon & \ldots & 0 & 0 \\
\vdots & \vdots & \vdots& \vdots &\ddots & \vdots & \vdots \\
1 & 0 & 0 & 0 & \ldots & 1-\epsilon & 0 \\
1 & 0 & 0 & 0 & \ldots & 0& 1-\epsilon \\
0 & 1 & 1 & 1 & \ldots & 1 & 1 
\end{bmatrix}.
\end{align*}
A key property of this instance is that the top $n-1$ agents are
ordinally indistinguishable.  Each of them ranks item~1 first, one of items
$\{2,\ldots,n\}$ second, and all other items last.  On the other hand,
each item $j \in \{2,\ldots,n\}$ is ordinally indistinguishable.  Each
is ranked second by exactly one of the top $n-1$ agents and
ranked equivalently by agent $n$.

Fix an ordinal mechanism. The ordinal indistinguishablity of agents
$\{1,\ldots,n-1\}$ implies, without loss of generality up to agent
relabeling, that agent~1 receives item~1 with probability at most
$1/(n-1)$.  Thus, in the limit of $\epsilon$ going to $0$, agent~1
obtains a utility of $1/(n-1)$ in this ordinal mechanism.

The Nash bargaining solution is continuous in $\epsilon$ and with
$\epsilon=0$ it gives each agent the maximum utility of 1 by
allocating item 1 to agent 1, item 2 to agent $n$, and item $i+1$ to
agent $i$ for $i \in \{2,\ldots,n-1\}$. Thus, in the limit, as
$\epsilon$ goes to zero the Nash bargaining solution gives agent 1 a
utility of 1. Therefore, the per-agent approximation of any ordinal 
mechanism with respect to the Nash bargaining benchmark is $n-1$.
\end{proof}

\section{Randomized Partial Improvement}
\label{s:mechanism}

In this section, we define the random partial improvement 
matching mechanism.  This mechanism truthfully elicits the agents'
cardinal preferences and uses them in a non-trivial manner to select
an outcome. We prove that the per-agent approximation of this
mechanism with respect to the Nash bargaining benchmark is
proportional to the population monotonicity of the benchmark and its
worst-case approximation is significantly better than that of any ordinal mechanism.
The approach of the mechanism is to run the PA mechanism with the outside
option given by the uniform random assignment on a large sample of the
agents and a large fraction of the supply.  The resulting mechanism
inherits the truthfulness of the PA mechanism.

There are two key difficulties with this approach.  First, in order to
faithfully implement the outside option, some of the supply needs to
be kept aside in the same proportion as the original supply.  
To enable this set aside, we need to reduce the allocation consumed
by the PA mechanism;
we achieve this with a novel use of random sampling \citep[cf.][]{GHKSW-06}. 
Second, it is non-trivial to compare
an agent's utility across Nash social welfare maximizing assignments
for the original market and a sample of the market.  A major endeavor
of our analysis shows that per-agent utility is approximately
monotone, i.e., the fraction of an agent's utility that is lost 
as the competition from other agents decreases is non-trivially bounded. 
(Note, competition from other agents decreases
as they are removed from the market.) Our mechanism, then, is
structured to take advantage of this approximate monotonicity.

The mechanism is defined by a sequence of steps that gradually
construct a doubly-stochastic matrix. By the Birkhoff von Neumann
theorem, this matrix corresponds to a probability distribution over
matchings.  The high-level steps and intuition are as follows: the
mechanism samples half the agents and runs at half scale (i.e., with
half-unit-demand agents and half-unit-supply items) the PA mechanism
with the outside option given by the uniform random assignment. The
total demand of half the agents (roughly $n/2$) with half-unit demand
is a quarter of the total supply (roughly $n/4$), so there is a
leftover $n/4$ of supply from the half-units on which
the PA mechanism was run.  A further one-quarter of each of the $n$
units is used to provide a half-unit of the outside option to each of
the (roughly) $n/2$ agents in the sample.  The final quarter is used
to replace as necessary the fractions of items withheld due to the fractional
reduction in the PA mechanism.  Necessarily, the one-unit allocation
to these agents uses up half the supply.  The remaining half
of the supply is then allocated recursively to the remaining half of
the agents.  

A formal description of this mechanism is below. Since
we call this mechanism recursively after some agents' allocation has been 
finalized (in the form of marginal probabilities) and some portions 
of the items have been allocated, we define it for the remaining $\bar{n}\leq n$
agents and the original $n$ items whose capacities may have been reduced from 1
to $(c_i)_{i\in M}\in [0,1]^n$.

\begin{definition}
Given some value $n_0\in \mathbb{N}$, the {\em randomized partial improvement} (RPI) 
mechanism on $\bar{n} \leq n$ agents and $n$ items with 
supplies $c_1,\ldots,c_n$ such that $\sum_{j=1}^n c_j = \bar{n}$ works as follows:
\begin{enumerate}
\item If $\bar{n}<n_0$, allocate the remaining item capacities uniformly at random, i.e.,
return $\alloc_{i,j} = c_j/\bar{n}$ for each agent $i$ and terminate. Otherwise, continue.
\item Randomly sample a subset $N'$ of $n' = \ceil{\bar{n}/2}$ agents.
\item On the sampled agents run the PA mechanism with the
  outside option given by the uniform random assignment $o'_i=\frac{1}{\bar{n}}\sum_{j=1}^n v_{i,j}c_j$ 
	from the supplies. Denote the allocation of item $j$ to agent $i$ by
  $\pa'_{i,j}$; the total amount allocated to agent $i$ is $f'_i = \sum_{j=1}^n \pa'_{i,j}$.
\item Allocate to each $i\in N'$ half of their PA assignment
  and ``pad'' it with the outside option to ensure a unit allocation. 
	As a result, the total allocation of item $j$
  to agent $i$ is $\alloc'_{i,j} = \pa'_{i,j}/2 + (1-f'_i/2)\, c_j / \bar{n}$.
\item Recursively run RPI on the remaining $n'' = \bar{n}-n'$ agents and item supplies 
$c''_j = c_j - \sum_{i\in N'} \alloc'_{i,j}$.
\item Return the assignment $\alloc$ that combines the assignment $\alloc'$ for 
$N'$ and the assignment $\alloc''$ returned by the recursive call for the remaining $n''$ agents.
\end{enumerate}
\end{definition}


The following proof of correctness (feasibility and truthfulness)
formalizes the intuition preceding the definition of the mechanism.
The ideas of the proof are more transparent in the case where $\bar{n}$ is
even and, in particular, $\ceil{\bar{n}/2} = \floor{\bar{n}/2} = \bar{n}/2$.

\begin{theorem}
The randomized partial improvement mechanism with $n_0 \geq 4$ on $n$
unit-demand agents and $n$ unit-supply items is feasible, i.e., it gives
fractional allocations that produce a doubly stochastic matrix, and
truthful, i.e., it is a dominant strategy for each agent to truthfully
report her value for each item.
\end{theorem}

\begin{proof}
Feasibility is proved by induction on the recursive definition of the
mechanism.  The inductive hypothesis is that the fractional allocation
on $\bar{n}$ agents with supplies $c_1,\ldots,c_n$ such that $\sum_j c_j = \bar{n}$
has a total fractional allocation to each agent of one, i.e., $\sum_j
\alloc_{i,j} = 1$ for each $i$, and a total fractional allocation of
each item equal to its supply, i.e., $\sum_i \alloc_{i,j} = c_j$ for
each $j$.  The base case of $\bar{n} < n_0$ clearly satisfies the inductive
hypothesis.  For the inductive step, the key point to argue is that
the supply $c_j$ of each item $j$ is sufficient to cover
the allocation to the sampled agents. 

This can be seen as follows.  The $\ceil{\bar{n}/2}$ sampled agents are
allocated half of their PA assignment on the supplied capacities. 
Since $f'_i \ge \frac 1e$ for all $i$ by Theorem~\ref{thm:CGG13}, this
means that the amount of each agent's half PA assignment is 
$\sum_j \pa'_{i,j}/2 = f'_i/2 \geq 1/(2e)$. To ensure that each agent gets
exactly one item in expectation, Step 4 pads this allocation with a uniformly
random assignment, which may thus require up to $1-1/(2e)$ units for each of
the $\ceil{\bar{n}/2}$ sampled agents, for a total of $(1- 1/(2e)) \ceil{\bar{n}/2}$. 
But, since the full PA assignment allocated no more than $c_j$ of each item $j$,
the half PA assignment set aside at least $c_j/2$ of each item, leading to a total
of $\sum_j c_j/2 = \bar{n}/2$. We conclude by observing that $(1- 1/(2e)) \ceil{\bar{n}/2}$ 
is at most $\bar{n}/2$ when $\bar{n} \geq n_0 = 4$, so the amount set aside from each
item is sufficient to cover the sampled agents' allocation.



Truthfulness follows by considering each agent conditioned on the
state of the mechanism during the recursive step where that agent is
selected in the sampled set $N'$.  The agent's report plays no role in
determining the state at this point.  Given the state, the outcome for
this agent is fully determined by the PA mechanism
which is truthful.  Thus, the mechanism is truthful.
\end{proof}

To bound the per-agent utility of the RPI
mechanism, we analyze the contribution to the utility of an agent who
is sampled in the outermost recursive call of the mechanism.  An agent
is sampled as such with probability at least one half, and otherwise
the agent's utility is at least zero.  The utility of these sampled
agents is easily compared to the utility of the PA mechanism (without
the agents that are not sampled).  An issue significantly complicating
the analysis of the approximation is the fact that we need to compare
the utility of an agent sampled in this invocation of the PA mechanism
with their utility in the Nash bargaining solution on the full set of agents.  
Counterintuitively, it is not true that these agents are always better off without the competition from the
agents that are not sampled: there are instances where removing some
of the competition, in fact, lowers the utility of an agent.

In Section~\ref{s:monotonicity} we define the {\em
  $\utilapprox$-utility monotonicity for NSW} to be the maximum
non-monotonicity of utility of any agent and sets of agents $N$ and
subset $N'$ with NSW maximizing solutions $p$ and $p'$ respectively:
\[ \utilapprox :=\max_{N'\subseteq N}\left\{\max_{i\in N'}\left\{\frac{\sum_j v_{i,j}p_{i,j}}{\sum_j v_{i,j}p'_{i,j}}\right\}\right\}.\]
This parameter quantifies the extent to which some agent may be worse
off in the NSW maximizing solution after the removal of some subset of agents.
Defining the worst-case value of $\utilapprox$ across instances and
subsets as $\wcutilapprox$, Section~\ref{s:monotonicity} bounds
$\wcutilapprox$ between $\Omega(2^{\sqrt{\log n}/2})$ and
{\UpperBoundShort}, the latter of which is $o(n^{\epsilon})$ for any
constant $\epsilon>0$. It is worth noting that we ran experiments 
on a large set of instances and found that this value was actually no more 
than 1 in all of these instances. 

\begin{theorem}
The randomized partial improvement mechanism with $n_0 = 4$ on $n$
unit-demand agents and $n$ unit-supply items is a $4\,e\,\utilapprox$ 
approximation to the Nash bargaining solution with
disagreement point given by the uniform random assignment.
\end{theorem}

\begin{proof}
If $n < n_0 = 4$ then the base case of RPI is invoked and a uniform
assignment is returned.  With $n<4$, however, this assignment is a 
$3 < 4 \, e \, \utilapprox$ approximation, as each agent obtains 1/3 
of each item.

Otherwise, we analyze the contribution to the utility of an agent
conditioned on the agent being sampled in the first recursive call of
the algorithm. This event happens with probability at least $1/2$.  
When this happens the utility of the agent is half the utility of PA 
on the sampled agents plus half the utility from the outside option.  The
$\utilapprox$-utility monotonicity property implies that the utility
of an agent in the NSW maximizing outcome on the sample is a $\utilapprox$ 
approximation to the same agent's utility in the NSW maximizing outcome 
on the full set of agents. Running PA guarantees an $e$ fraction of this 
utility. Combining these steps and the fact that agents who are not sampled
in the first recursive call still receive nonnegative utility, we obtain a 
$4\,e\,\utilapprox$ approximation.
\end{proof}

Combined with Theorem~\ref{thm::upper-bdd-on-rho} from Section~\ref{s:monotonicity} which bounds $\wcutilapprox$ by \UpperBoundShort, we have the following.

\begin{cor}
The randomized partial improvement mechanism with $n_0 = 4$ on $n$
unit-demand agents and $n$ unit-supply items guarantees an approximation
of the Nash bargaining solution with uniform outside option with a factor \UpperBound.
\end{cor}

\section{Approximate Utility Monotonicity}\label{s:monotonicity}

A factor significantly complicating the analysis of the approximation
of the random partial improvement mechanism is the fact that the
benchmark is computed based on the Nash social welfare maximizing
solution when all agents in $N$ are present, while the mechanism's performance  depends on the solution for $N'$, the sampled agents.  The NSW
maximizing solution for $N$ and $N'$ can generally be quite different.
Moreover, as it turns out, there are instances where the utilities of
some agents in the NSW maximizing solution are non-monotone with respect to
removal of other agents, i.e., there exist instances that exhibit positive
externalities between agents.  Table~\ref{tbl::simple-ex-explained}
gives a simple example of such an instance (discussed in detail later on) and the 
remainder of the section develops upper and lower bounds on the worst-case 
non-monotonicity of utility.

\newcommand{\colspace}{0.8cm}
\begin{table*}[t]
\centering
\begin{tabular}{C{\colspace} C{\colspace} C{\colspace} C{\colspace} C{\colspace} C{\colspace} C{\colspace} C{\colspace} C{\colspace} C{\colspace} C{\colspace} C{\colspace}}
&$A$&$B$&$C$&~~~~& $A$ & $B$ & $C$ &~~~~& $A$ & $B$ & $C$ \\ 
\cline{2-4}      \cline{6-8}              \cline{10-12}
$a$ & 1 & 2 & 0  && 1 & 0 & 0           && $1/2$ & $1/2$ & 0     \\
$b$ & 0 & 2 & 1  && 0 & 1 & 0           && 0     & $1/2$ & $1/2$ \\
$c$ & 0 & 0 & 1  && 0 & 0 & 1                                    \\
\\
&\multicolumn{3}{c}{i. Agent valuations} && \multicolumn{3}{c}{ii. Initial solution} && \multicolumn{3}{c}{iii. Final solution} 
\end{tabular}
\caption{\label{tbl::simple-ex-explained} A simple instance involving
  three agents; $a,b,c$; and three items; $A,B,C$. The value of
  agent $b$ in the Nash bargaining solution after the removal of agent
  $c$ drops by a factor of $4/3$.  The optimality of solutions (ii)
  and (iii) for Nash social welfare are intuitive; however, formal justification is
  given in Section~\ref{subsec:LowerBound} where this example is
  revisited.}
\end{table*}



%
\begin{definition}\label{def:rho}
A matching environment on agents $N$ is {\em $\utilapprox$-utility
  monotone} if for any subset $N'$ of $N$ and any $i \in N'$ the
utility of $i$ in the NSW maximizing assignment, $p'$, 
for $N'$ is at least a $\utilapprox$ approximation to the NSW maximizing
assignment, $p$, for $N$:
\[ \utilapprox :=\max_{N'\subseteq N}\left\{\max_{i\in N'}\left\{\frac{\sum_j v_{i,j}p_{i,j}}{\sum_j v_{i,j}p'_{i,j}}\right\}\right\}.\]
\end{definition}  

This parameter quantifies the extent to which some agent may be worse
off in the NSW solution after the removal of some subset of agents. We
let $\rho^*$, denote the worst case value of $\rho$ across instances;
this value is known as the \emph{opportunity structure} of the Nash
bargaining solution for this class of instances~\citep{TL-89}.  In
Section~\ref{subsec:UpperBound} we prove an upper bound of
\UpperBoundShort, which is $o(n^\epsilon)$ for any constant $\epsilon
> 0$, for the value of $\rho^*$ over all one-sided matching instances,
and in Section~\ref{subsec:LowerBound} we complement this result by
proving a lower bound of $\Omega(2^{\sqrt{\log n}/2})$ for this value.

\subsection{Upper Bound.}\label{subsec:UpperBound}
Given a valuation matrix $v$ and a random matching $p$, we henceforth use $u_i(p)$ to denote the expected
utility of agent $i$ for $p$ given $v$, i.e., $\sum_{j\in M}v_{i,j}p_{i,j}$ (similarly, we use $u'_i(p)$
for valuation matrix $v'$). Note that, as we discussed in Section~\ref{s:benchmark}, the Nash bargaining
solution is scale invariant. Therefore, if we scale the valuations of each agent $i$ by some constant $c_i>0$,
then the Nash bargaining solution with respect to valuations $c_i v_{i,j}$ instead of $v_{i,j}$ will remain
the same. This means that given some problem instance that yields a doubly stochastic matrix $p$ as its Nash
bargaining solution, we can always ``normalize'' the valuations of the agents so that every agent's expected 
utility for $p$ is equal to 1, and $p$ remains that Nash bargaining solution of the normalized instance. This 
is a convenient normalization that we make use of below. 

In order to prove the upper bound on $\rho^*$, we first prove the following very useful lemmata.

\begin{lemma}\label{lem:value_dif}
Let $p$ be a NSW maximizing solution, and $v$ be the valuations normalized so that for every agent $i$, 
$u_i(p)=1$. Then, if some agent $i$ is allocated an item $j$ with positive probability, i.e., $p_{i,j}>0$, 
every other agent $k\neq i$ must have $v_{k,j} \leq v_{i,j}+1$. Equivalently, $v_{i,j}\geq \max_{k\in N} \{v_{k,j}\}-1$.
\end{lemma}
\begin{proof}
For contradiction, assume that there exist two agents $k$ and $i$ and an item $j$ such that 
$p_{i,j}>0$ and $v_{k,j} = v_{i,j}+1+\delta$ for some $\delta>0$. Since the expected utility of agent $k$ 
is 1, there must also exist some item $\ell$ with $p_{k,\ell}>0$ and $v_{k,\ell}\leq 1$ (otherwise the 
expected utility of agent $k$ would be greater than 1). Note that $\ell\neq j$, since $v_{k,j} = v_{i,j}+1+\delta>1$,
whereas $v_{k,\ell}\leq 1$. 

Let $p'$ be a probability distribution identical to $p$, 
except $p'_{k,j}=p_{k,j}+\epsilon$, $p'_{i,j}=p_{i,j}-\epsilon$, $p'_{k,\ell}=p_{k,\ell}-\epsilon$,
and $p'_{i,\ell}=p_{i,\ell}+\epsilon$, for some positive $\epsilon <\min\{p_{i,j}, p_{k,\ell}\}$,
whose exact value we choose later on.
In other words, $p'$ swaps probability $\epsilon$ between agents $i, k$ and items $j, \ell$.
The expected utility of agent $k$ in $p'$ is 
\begin{align*}
u_k(p') &~=~ 1+\epsilon v_{k, j}-\epsilon v_{k, \ell} \\
& ~\geq~ 1+\epsilon(v_{i,j}+1+\delta) -\epsilon \\
& ~=~ 1+ \epsilon v_{i,j}+\epsilon\delta,
\end{align*}
\noindent and the expected utility of agent $i$ in $p'$ is
\begin{align*}
u_i(p') &~=~ 1+\epsilon v_{i, \ell}-\epsilon v_{i, j} \\
& ~\geq~ 1-\epsilon v_{i, j}.
\end{align*}
Since every other agent's expected utility is the same in $p'$ and $p$ (equal to 1), the NSW of $p'$ is
\begin{align*}
\prod_{i\in N} u_i(p') & ~\geq~ (1+ \epsilon v_{i,j}+\epsilon\delta)(1-\epsilon v_{i, j})\\
&  ~=~ 1 + \epsilon(\delta - \epsilon(v_{i,j}^2+\delta v_{i,j})).
\end{align*} 
Therefore, if we let $\epsilon < \delta/(v_{i,j}^2+\delta v_{i,j}))$, the NSW of $p'$ is greater than 1, which is the NSW of $p$,
contradicting the fact that $p$ is a NSW maximizing solution.
\end{proof}

\begin{lemma}\label{lem:ratio_change}
Given a problem instance, let $p$ and $\overline{p}$ be the NSW maximizing outcomes before and after (respectively) 
some subset of the agents has been removed. If among the remaining agents there exists a set of agents $N_1$ and a constant 
$d\geq 12$ such that every agent $i\in N_1$ has $u_i(\overline{p})\leq u_i(p)/d$, then there 
also exists a larger set $N_2$ of remaining agents such that $|N_2|\geq d|N_1|/3$ and for all agents $k\in N_2$ we have 
$u_k(\overline{p})\leq  4 u_k(p)/d$.
\end{lemma}
\begin{proof}
Without loss of generality, let $v$ be the agent valuations normalized so that $u_i(p)=1$ for every agent $i$, 
and $v'$ be the valuations normalized so that $u'_i(\overline{p})=1$. Given the $v_{i,j}$ values that yield
$u_i(p)=1$, we can get the $v'_{i,j}$ values that yield $u'_i(\overline{p})=1$ using the simple formula
$v'_{i,j} = v_{i,j} \cdot\frac{u_i(p)}{u_i(\overline{p})}$. In other words, for each agent $i$ who is worse-off 
in $\overline{p}$ compared to $p$, i.e., $u_i(\overline{p})< u_i(p)$, we scale all of that agent's item values
up by the same factor, $u_i(p)/u_i(\overline{p})$. In particular, for each agent $i\in N_1$ this means that 
$v'_{i,j}\geq d v_{i,j}$ for every item $j$.

%

For every $i\in N_1$ we know that the drop in that agent's value with respect to the original valuations is 
$u_i(p)-u_i(\overline{p})\geq u_i(p)(1-1/d) = 1-1/d$.
In order to account for that drop, we partition the set of items of which $i$ is allocated more in $p$ compared to $\overline{p}$ 
into two sets depending on whether $v_{i,j}\geq 0.5$
or not: $M_i^h=\{j\in M: p_{i,j}> \overline{p}_{i,j} \text{ and }  v_{i,j}\geq 0.5\}$ and 
$M_i^{\ell}=\{j\in M: p_{i,j}> \overline{p}_{i,j} \text{ and }  v_{i,j} < 0.5\}$.
We first show that from the aforementioned $1-1/d$ drop in value, no more than $0.5$ could be due to the items in $M_i^{\ell}$, since 
\begin{equation*}
\sum_{j\in M_i^{\ell}} (p_{i,j}-\overline{p}_{i,j})v_{i,j} ~<~ 0.5\sum_{j\in M_i^{\ell}} (p_{i,j}-\overline{p}_{i,j})  ~\leq~ 0.5.
\end{equation*}
Therefore, at least $0.5-1/d$ of this drop in value for each agent $i\in N_1$ is due to items in $M_i^h$. Summing this up
over all the agents in $N_1$, we get
\begin{align}\label{eq:drop_bound}
\sum_{i\in N_1}\sum_{j\in M_i^h} (p_{i,j}-\overline{p}_{i,j})v_{i,j} & ~\geq~ \sum_{i\in N_1} \left(0.5-\frac{1}{d}\right) \\
& ~=~ |N_1|\left(0.5-\frac{1}{d}\right).
\end{align}

Let $N_2=\{k\in N: \overline{p}_{k,j}>0 \text{ for some } j\in \bigcup_{i\in N_1} M_i^h\}$ be the set of agents that
are allocated with positive probability in $\overline{p}$ an item from $M_i^h$ for some $i\in N_1$. Using 
Lemma~\ref{lem:value_dif} we get that for every item $j\in M_i^h$, if $\overline{p}_{k,j}>0$ then $v_{k,j}\leq v_{i,j}+1$
and $v'_{k,j}\geq v'_{i,j}-1$. Using the fact that $v'_{i,j} \geq d v_{i,j}$ for every $i\in N_1$, shown above, the latter
inequality also implies that $v'_{k,j} \geq d v_{i,j}-1$. Therefore 
\[\frac{v'_{k,j}}{v_{k,j}} ~\geq~ \frac{d v_{i,j}-1}{v_{i,j}+1}  ~\geq~ d - \frac{d+1}{1.5} ~\geq~ \frac{d}{4},\]
where the second inequality uses the fact that $v_{i,j}\geq 0.5$ and the last inequality uses the fact that $d\geq 8$.
This implies that for every $k\in N_2$
\begin{align*}u_k(\overline{p}) &~=~ \sum_{j\in M} \overline{p}_{k,j} v_{k,j}\\ 
&~\leq~ \sum_{j\in M} \overline{p}_{k,j} \frac{4}{d}v'_{k,j}\\
&~=~ \frac{4}{d}\sum_{j\in M} \overline{p}_{k,j} v'_{k,j}\\
&~=~  \frac{4 u'_k(p)}{d} \\
& ~=~ \frac{4 u_k(p)}{d}, 
\end{align*}
\noindent where the last equation uses the fact that $u'_k(p)=u_k(p)=1$ according to our normalization.

Since we have shown that for all $k\in N_2$ we have $u_k(\overline{p})\leq  4 u_k(p)/d$,
it now suffices to show that the size of $N_2$ is at least $d|N_1|/3$.
Since, for any item $j\in \bigcup_{i\in N_1} M_i^h$, any agent $k$ with $\overline{p}_{i,j}>0$ satisfies 
$v'_{k,j}\geq  d v_{i,j}-1$, the total value, with respect to valuations $v'$, generated by the item fractions 
of the items ``lost'' by the agents in $N_1$ is at least
\begin{align*}
\sum_{i\in N_1} \sum_{j\in M_i^h} &(p_{i,j}-\overline{p}_{i,j}) (d v_{i,j}-1) \\
~&\geq~ d|N_1|\left(0.5-\frac{1}{d}\right) -\sum_{i\in N_1}\sum_{j\in M_i^h} (p_{i,j}-\overline{p}_{i,j})\\
~&\geq~ d|N_1|\left(0.5-\frac{1}{d}\right) -|N_1|\\
~&\geq~ \frac{d-4}{2}|N_1|\\
~&\geq~ \frac{d}{3}|N_1|,
\end{align*}
where the last inequality uses the fact that $d\geq 12$. But, since the total value of each agent in $N_2$
with respect to valuations $v'$ is exactly 1, there need to be at least $\frac{d}{3}|N_1|$ agents in $N_2$
sharing this value, otherwise there would exist some agent $i\in N_2$ such that $u'_i(\overline{p})>1$. 
\end{proof}

\begin{theorem}
\label{thm::upper-bdd-on-rho}
For any problem instance, the value of $\rho$ is \UpperBound.
\end{theorem}
\begin{proof}
In order to prove this bound, we will repeatedly apply the result of Lemma~\ref{lem:ratio_change}.
Let $p$ and $\overline{p}$ be the NSW maximizing outcomes in a problem instance before and after 
some subset of the agents has been removed and, without loss of generality, let $v$ be the agent 
valuations normalized so that $u_i(p)=1$ for every agent $i$, and $v'$ be the valuations normalized so 
that $u'_i(\overline{p})=1$. 

By Definition~\ref{def:rho}, in an instance with utility monotonicity equal to $\rho$, there exists 
at least one agent $i\in N_1$ such that $u_i(p)/u_i(\overline{p})=\rho$ or $u_i(\overline{p})=u_i(p)/\rho$.
If $\rho > 12$, then Lemma~\ref{lem:ratio_change} would imply that there also exists a set $N_2$ of at least 
$\rho/3$ agents such that $u_k(\overline{p})\leq 4 u_k(p)/\rho = 4/\rho$ for every $k\in N_2$. Lemma~\ref{lem:ratio_change}, 
combined with the existence of the set $N_2$, in turn, implies the existence of an even larger group $N_3$ of at least 
$(\frac 13 \rho)\cdot (\frac 13 \rho/4)$ agents, and each agent $k\in N_3$ has value $u_k(\overline{p})\leq 16/\rho$. Applying Lemma~\ref{lem:ratio_change} 
a total of $\alpha$ times 
thus implies the existence of a set of at least $(\rho/3)^{\alpha}\cdot (1/4)^{\alpha(\alpha -1)/2}$ agents such that each such agent $k$
has value $u_k(\overline{p})\leq 4^{\alpha}/\rho$. 
Assume that there exists some instance for which $\rho$ is at least $4^{\sqrt{\log n}+1}$. If we choose $\alpha=\sqrt{\log n}$, 
however, this implies the existence of 
$(\rho/3)^{\alpha}\cdot (1/4^{\alpha(\alpha -1)/2})
\ge 4^{(\sqrt{\log n}+1)\cdot \sqrt{\log n}}/[(3/2)^{\sqrt{\log n}}\cdot 2^{\sqrt{\log n}\cdot \sqrt{\log n}}]\ge (8/3)^{\sqrt{\log n}}\cdot n \ge n$ 
agents of value at most 
$4^{\alpha}/\rho \le 1/4$. 
But, this would imply that all the agents have a value less
than 1 in $\overline{p}$, which contradicts the fact that $\overline{p}$ is a NSW maximizing solution because
the product in $p$ is equal to 1.
\end{proof}

\subsection{Lower Bound.}\label{subsec:LowerBound}

We conclude with a lower bound showing that for a very carefully designed (and somewhat artificial) family of 
instances, the upper bound of Theorem~\ref{thm::upper-bdd-on-rho} is essentially tight. 
\begin{theorem}
\label{thm::lower-bdd-on-rho}
There exists a family of problem instances for which $\rho^* = \Omega(2^{\sqrt{\log n}/2})$.
\end{theorem}

Due to space limitations and the complexity of the construction that
yields Theorem~\ref{thm::lower-bdd-on-rho}, we defer its description
to Appendix~\ref{app:lower-bound}. To exhibit how we use KKT conditions to prove that
this elaborate construction implies the desired bound, we use the rest
of this section to apply this approach to the much simpler
construction of the example in Table~\ref{tbl::simple-ex-explained},
which yields a bound of $\rho^*\geq 4/3$.

Our lower bound construction in the appendix proceeds by building a family of instances (parameterized by the number
of agents $n$), and in each instance, we define an ``initial'' setting in which all agents are present, and a ``final'' 
setting, in which some agents have been removed. For each setting, we identify the Nash bargaining solution, 
respectively called the initial and final solution. We focus on a particular agent, called the \emph{loser}, 
who is present in both settings. We show that the loser's valuation drops by a multiplicative factor $\mu$
in going from the initial to the final solution, and consequently, $\rho \ge 1/\mu$ for that market and
$\rho^* \ge 1/\mu$ overall. 

To prove a lower bound, we need to be able to verify that a given doubly stochastic matrix is indeed
the Nash bargaining solution of the instance at hand. We do so using the KKT conditions, which allow
us to interpret these solutions as a form of market equilibrium. The optimization problem which yields 
the Nash bargaining solution in one-sided matching markets, is shown below (where $m=n$ is used to denote 
the number of items):
\begin{align*}
& \max \sum_{i=1}^n \log \left[\sum_{j=1}^n \vij \pij \right] \\
& ~~\text{such that: ~ for all} ~ i:~~\sum_{j=1}^m \pij \le 1\\
&~~~~~~~~~~~~~~~~~~~\text{for all} ~ j:~~\sum_{i=1}^n \pij \le 1\\
&~~~~~~~~~~~~~~~~~~~\text{for all} ~i,j:~~\pij \ge 0.
\end{align*}

If $t_j$ is the dual variable related to each item $j$, and $q_i$ is the dual variable related to each
agent $i$ in the above program, then the KKT conditions state that:
\begin{align}
\label{p-q-constraints}
&\text{for all}~j:~t_j \ge 0,  ~\text{ and }~  t_j >0 \Longrightarrow \sum_{i=1}^n \pij  = 1\\
\label{q-positive}
&\text{for all}~i:~q_i \ge 0  ~\text{ and }~ q_i > 0 \Longrightarrow \sum_{j=1}^m \pij = 1\\
\label{tight-utility-rule}
\begin{split}
& \text{for all} ~i,j: ~~ \frac {\vij}{t_j+q_i} \le \sum_{j=1}^m \vij \pij  ~\text{ and }\\
 &~~~ \pij > 0 \Longrightarrow \frac {\vij}{t_j+q_i} = \sum_{j=1}^m \vij \pij
\end{split}
\end{align}

The KKT conditions are necessary and sufficient for the optimal solution when the constraints are linear and the objective 
is convex, as is the case here. To check whether a given candidate solution $p$ is a Nash bargaining solution for some instance, 
we first normalize the valuations so that $\sum_{j=1}^m \vij \pij = 1$ for all $i$. Then, at a solution satisfying the KKT conditions 
we have $\vij = t_j+q_i$ if $p_{i,j}>0$ and $\vij \le t_j + q_i$ if $p_{i,j}=0$. Thus a solution that satisfies these two conditions plus 
conditions~\eqref{p-q-constraints}--\eqref{q-positive} is a Nash bargaining solution. Based on this conditions, the values of $t_j$ can
be interpreted as item-specific ``prices'' and the values of $q_i$ as agent-specific ``prices'', leading to an interpretation of the Nash 
bargaining solution as a market equilibrium: to ``buy'' a $p_{i,j}$ fraction of item $j$, agent $i$ needs to spend 
$(t_j+q_i)p_{i,j}$, and each agent prefers to buy only items with the best value over price ratio (see condition \eqref{tight-utility-rule}).

To illustrate the usefulness of these variables, which are used extensively in the appendix,
we revisit the instance of Table~\ref{tbl::simple-ex-explained} where the items are named $A$, $B$, and $C$; the bidders 
$a$, $b$, and $c$; and the unscaled valuations of the agents appear in Table~\ref{tbl::simple-ex}(i). 

First, we observe that in the initial equilibrium (with all agents present), $a,b,c$ receive items $A,B,C$, respectively,
each with probability 1. In Table~\ref{tbl::simple-ex}(ii) we show the normalized values of the agents in this equilibrium
and we also provide the dual variables $t_j$ for each item $j$ and $q_i$ for each agent $i$. It is easy to verify that the
aforementioned KKT conditions are satisfied in this case and hence this is indeed the Nash bargaining solution when all agents
are present. If agent $c$ is removed, then the final equilibrium finds $a$ receiving each of $A$ and $B$ with probability $\frac 12$, 
while $b$ receiving each of $B$ and $C$ with probability $\frac 12$. Table~\ref{tbl::simple-ex}(iii) provides the scaled valuations
and dual variable values for this outcome, and it is again easy to verify that KKT conditions are satisfied. 
In this example, bidder $b$ is the loser. Using the valuations from Table~\ref{tbl::simple-ex}(ii), her value in the initial
equilibrium was 1 and it dropped to $0.75$ in the final equilibrium, leading to $\rho = \frac 43$ in this 
example.

%

\renewcommand{\colspace}{0.8cm}
\newcommand{\tabspace}{0.1cm}
\newcommand{\tspace}{0.2cm}
\begin{table*}[t]
\centering
\begin{tabular}{C{\tspace} C{\colspace} C{\colspace} C{\colspace} C{\tabspace} C{\tspace} C{\colspace} C{\colspace} C{\colspace} C{\colspace} C{\tabspace} C{\tspace} C{\colspace} C{\colspace} C{\colspace} C{\colspace}}
&$A$&$B$&$C$ &  && $A$ & $B$ & $C$ & $q$ &       && $A$ & $B$ & $C$ & $q$ \\ 
\cline{2-4}      \cline{7-9}                     \cline{13-15}
$a$ & 1 & 2 & 0 && & {\bf 1} & 2 & 0     & 1     && & \boldmath{$2/3$} & \boldmath{$4/3$} & 0 & 0\\
$b$ & 0 & 2 & 1 && & 0 & {\bf 1} & $1/2$ & 0     && & 0 & \boldmath{$4/3$} & \boldmath{$2/3$} & 0 \\
$c$ & 0 & 0 & 1 && & 0 & 0 & {\bf 1}     & 0     && & & \\ 
\cline{2-4}      \cline{7-9}                      \cline{13-15}
\\ & & &           &&$t$& 0 & 1 & 1 &               &&$t$ & $2/3$ & $4/3$ & $2/3$ & \\
\\
& \multicolumn{3}{c}{i. Agent valuations} &&& \multicolumn{3}{c}{ii. Initial solution} &&&& \multicolumn{3}{c}{iii. Final solution} \\
\end{tabular}
\caption{\label{tbl::simple-ex}Simple instance involving three agents; $a,b,c$; and three items; $A,B,C$. The value of agent $b$ in the Nash bargaining solution
after the removal of agent $c$ drops by a factor of $4/3$.  Normalized values are depicted in (ii) and (iii) along with prices $t$ and $q$; these values are depicted in bold-face if the allocation probability of the NSW solution is non-zero.}
\end{table*}

\section{Further Related Work}\label{sec:related}

\citet{HZ-79} study the problem of matching with cardinal preferences
and the solution of competitive equilibrium from equal incomes (CEEI).
CEEI gives both a natural cardinal notion of efficiency and of
fairness.  Recently, \citet{AKT-17} give a polynomial time algorithm
for computing the CEEI in matching markets when there are a constant
number of distinct agent preferences.  To our knowledge, the
complexity of computing CEEI in general matching problems is unknown.
With linear preferences, but without the unit-demand constraint, CEEI
and Nash social welfare coincide and can be computed in polynomial
time.  \citet{dk} generalize this computational result to piecewise
linear concave utilities when the number of goods is constant.

Recently, \citet{bud-11} considers the generalization from matching
to a combinatorial assignment problem where agents may have non-linear
preferences over bundles of goods, and shows that an approximate version
of CEEI exists. This work also shows that, in large markets, the mechanism that outputs
this approximate CEEI is asymptotically truthful.  Heuristics for
computing the CEEI outcome are given by \citet{OSB-10} and these
heuristics have been deployed for the course assignment problem by
\citet{BCKO-16}.  On the other hand, \citet{OPR-16} show that the
computation of CEEI in these combinatorial assignment problems is
PPAD-hard.

The Nash social welfare objective of our work compares to competitive
equilibrium from equal incomes of the aforementioned works as
follows: the two objectives coincide for linear preferences without
the matching constraint \citep{vaz-07}, but with the matching
constraint the concepts are not equivalent.  Both NSW and CEEI
outcomes are Pareto efficient, but to our knowledge, in
matching markets, the agents' utilities under the two criteria have not
been directly compared.  Contrasting with CEEI, for stochastic
matchings, the NSW outcomes can be calculated by a convex program,
i.e., a program that optimizes the product of utilities over the marginal
probabilities given by a doubly-stochastic matrix, and is therefore
computationally tractable.

A second line of literature considers ordinal mechanisms for one-sided
matching.  The {\em random serial dictatorship} (RSD) mechanism has a long
history of practical application. Recently it has been used in
applications such as housing and course allocation.
\citet{PS-11} study the use of RSD for school choice in New York
City.  RSD is truthful, ex post Pareto efficient, and easy to
implement \citep[e.g.,][]{AS-98}.  On the other hand, RSD is neither
ex ante Pareto efficient nor envy-free.  To remedy this deficiency of
RSD, \citet{BM-01} developed the {\em probabilistic serial} (PS) mechanism
which, while not truthful, is ordinally efficient, envy-free,
and easy to implement.  PS has been studied in various contexts
ranging from school assignments to kidney matching and it is often
contrasted with RSD. For example, \citet{PS-11} show that students
often obtain a more desirable random assignment from PS than from RSD.
Nonetheless, under a large market assumption PS and RSD converge and
the desirable properties of both are attained
\citep{KM-10,CK-10}. More recent work has also further studied and
compared the performance of these two mechanisms with respect to different
metrics both theoretically and experimentally (e.g., \cite{AFCMM16,HLC18}).

Several recent papers have considered approximation in one-sided
matching markets without money when agents have cardinal preferences.
With cardinal preferences, it is possible to consider the aggregate
welfare of an allocation as the sum of the expected utilities of each
agent.  For an aggregate notion of welfare to make sense, the values
of the agents need to be normalized.  Two common normalizations are
unit-sum, which scales each agent's values so that their sum is one,
and unit-range, which scales and shifts each agent's values so that
the minimum value is zero and the maximum value is one.  Under either
of these normalizations, \citet{FFZ-14} show that randomized serial
dictatorship is an $\Theta(\sqrt{n})$ approximation and that no
algorithm for mapping ordinal preferences to allocations is
asymptotically better.  \citet{CFFGZ-16} consider the unit-sum
normalization and show that the price of anarchy of PS is
$\Theta(\sqrt{n})$ and that no mechanism, ordinal or cardinal, is
asymptotically better.  Important comparison of these above results to
ours are as follows: our guarantees do not require a normalization of
values.  Our approximation guarantees are on per-agent utilities, not
on the aggregate welfare which allows some agents to be harmed if
other agents benefit.  We show that our randomized partial improvement
mechanism is asymptotically better than RSD in our per-agent analysis
framework by a factor of $\Omega(\sqrt{n})$.

More recently, \citet{ILWM-17} use a notion of approximate Pareto efficiency 
to analyze the {\em raffles mechanism} in one-sided matching markets. This efficiency 
measure provides per-agent approximation guarantees with respect to the Pareto frontier. Our approximation 
measure can therefore be thought of as a refinement where instead we compare the agent utilities to 
a specific highly desired point on the Pareto frontier (the Nash bargaining solution). 
Instead of eliciting the preferences
of the agents, the raffles mechanism instead provides the agents with tickets that they can allocate 
to items, and items are distributed in proportion to the allocated tickets. As a result, this mechanism is not truthful, but the main result shows that its Nash equilibria are $e/(e-1)$-approximately Pareto efficient, i.e., that there is no equilibrium where each agent's utility is increased by more than an $e/(e-1)$ factor.

Our mechanism is based on the {\em partial allocation} (PA) mechanism of
\citet{CGG-13} that truthfully and approximately solves the fair
division of heterogeneous goods.  A novel feature of the PA mechanism
is that a fraction of the fair allocation is withheld from individual
agents in a way that behaves, in the agents' utilities, as payments
that align the incentives of the agents with the Nash social welfare
objective.  The fair division problem is closely tied to the
cake-cutting literature, which originated in the social sciences but
has garnered interest from computer scientists and mathematicians
alike \citep{bt,m,rw,y}. The cake -- a heterogeneous, divisible item -- 
is represented by the interval $[0,
  1]$ and the agents have valuation functions assigning each
subinterval to a non-negative value. These valuations are also assumed
to be additive. Algorithmic challenges in cake cutting have recently
attracted the attention of computer scientists. A historical overview
as well as notable results in cake cutting can be found in surveys
by \citet{pro-13} and \citet{PM-16}. The cardinal matching problem we 
consider is closely related to the cake
cutting problem with piecewise uniform valuations since our agents
have linear preferences over items.

Random sampling techniques are now common in the literature on mechanism
design.  They have been primarily developed for revenue maximization
problems where the seller lacks prior information on the agents'
preferences \citep{HK-07}.  Our use of random sampling more closely
resembles the literature on redistribution mechanisms, where the
designer aims to maximize the consumer surplus and monetary transfers
between agents are allowed \citep{cav-06,GV-07}.  An approach by
\citet{mou-09} is to single out a random agent as the residual
claimant, run an efficient mechanism on the remaining agents, and pay
the revenue generated by the mechanism to the residual claimant.
Similarly, our mechanism randomly partitions the agents into two
groups and attempts to implement the PA mechanism on
the first group while using the items that would be reserved for the
second group to implement the first group's outside option.  Further
connections between our approach and redistribution mechanisms may be
possible.

\section{Conclusion and Future Work}

We defined the random partial improvement (RPI) mechanism for
one-sided matching markets without monetary transfers. RPI both
truthfully elicits the cardinal preferences of the agents and outputs
a distribution over matchings that approximates every agent's utility
in the Nash bargaining solution.


Our analysis suggests several open questions and directions for future
work. A natural open question is whether there exists a truthful mechanism 
that can achieve a constant factor approximation of the Nash bargaining 
benchmark. The main obstacle for the RPI mechanism was the non-monotonicity 
of the Nash bargaining benchmark, so it would be interesting to see if some other mechanism
could circumvent this issue. Alternatively, since the construction leading
to the lower bound is quite artificial, are there any natural assumptions
regarding the valuations of the agents that would mitigate the non-monotonicity?

Another interesting direction would be to study how the utilities of agents in
the CEEI outcomes compare to those of the Nash bargaining solution. Recall that
the CEEI and the Nash bargaining solution are equivalent in linear markets without 
the matching constraint \citep{vaz-07}, but are different for matching
markets. 

Finally, our paper provides a non-trivial mechanism aiming to approximate a
well-motivated ex-ante Pareto efficient outcome. One could also consider the
design of truthful mechanisms aiming to approximate alternative benchmarks on 
the ex-ante Pareto frontier. Natural candidates would be the utilitarian (or
the egalitarian) outcome which maximize the sum (or the minimum) of the agents'
utilities. One drawback of these outcomes is that, unlike the Nash bargaining
solution, they are not scale invariant, but one could consider scaled variants
of their objectives, e.g., where the agent values are normalized so that 
$\sum_{j\in M} v_{i,j}=1$ for every agent $i$.






\bibliographystyle{apalike}
\bibliography{refs}

\appendix

\section{The Lower Bound Construction}
\label{app:lower-bound}

The construction uses a collection of overlapping submarkets named $M_0, M_1, \ldots, M_s$. Associated with these markets
are integer parameters $\kr$, for $1 \le r \le s$. There will be $\kr$ copies of $M_r$ in the construction. As we shall see, $k_s=1$, thus
 there will be exactly one copy of $\Ms$.

Next, in Figure~\ref{fig::M0-flow} and Table~\ref{tbl::Co}, we show the form of $M_0$. The nodes in this figure
correspond to the items, and a directed edge $(\alpha, \beta)$, labeled by the name of an agent, indicates that
this agent was allocated portions of item $\alpha$ in the initial solution and portions of $\beta$ in the final solution.
Items and bidders occur with multiplicity possibly greater than 1 and this is called their size.
In the initial equilibrium, every item is fully allocated as the total size of the bidders
and items are the same; in the final equilibrium, item $A_0$ is the one item
that is not fully allocated.
Note that in each equilibrium, the conditions \eqref{p-q-constraints}--\eqref{tight-utility-rule}
from Section~\ref{subsec:LowerBound} are satisfied.

\begin{figure*}[htbp]
\begin{center}

\begin{tikzpicture}
\path (0,0) node [rectangle,draw](a) {items $A_0$, size 17,000} --
(0,-1.5) node[rectangle,draw](b) {items $B_0$, size 850} --
(-2,-3)  node[rectangle,draw](c) {items $C_0$, size 816} --
(2,-3)  node[rectangle,draw](d) {items $D_0$, size 34} --
(-4, - 4.5) node[draw](cc) {exit in final equil.} --
(0,-4.5)  node[rectangle,draw](e) {items $E_0$, size 33} --
(0, - 6) node[draw](ee) {exit in final equil.} --
(4,-4.5)  node[rectangle,draw](f) {items $F_0$, size 1} --
(8,-4.5)  node[rectangle,draw](g) {items $G_0$, size 5, alloc.\ to $f_0$} --
(4, -6) node[rectangle,draw](h) {items $H_0$, size 1};
 \draw[->] (a) -- node[right] {$a_0$} (b);
 \draw[->] (b)  -- node[right] {$b_0$} (c);
 \draw[->] (b)  -- node[left] {$b_0$} (d);
 \draw[->] (c)  -- node[right] {$c_0$} (cc);
 \draw[->] (d)  -- node[right] {$d_0$} (e);
 \draw[->] (e)  -- node[right] {$e_0$} (ee);
 \draw[->] (d)  -- node[left] {$d_0$} (f);
 \draw[->] (f)  -- node[left] {$f_0$} (h);
 \draw[->] (h)  -- node[right] {$h_0= a_1$} (4,-7.5);
  
\end{tikzpicture}
\caption{\label{fig::M0-flow}The Allocations in Market $M_0$.}
\end{center}
\end{figure*}
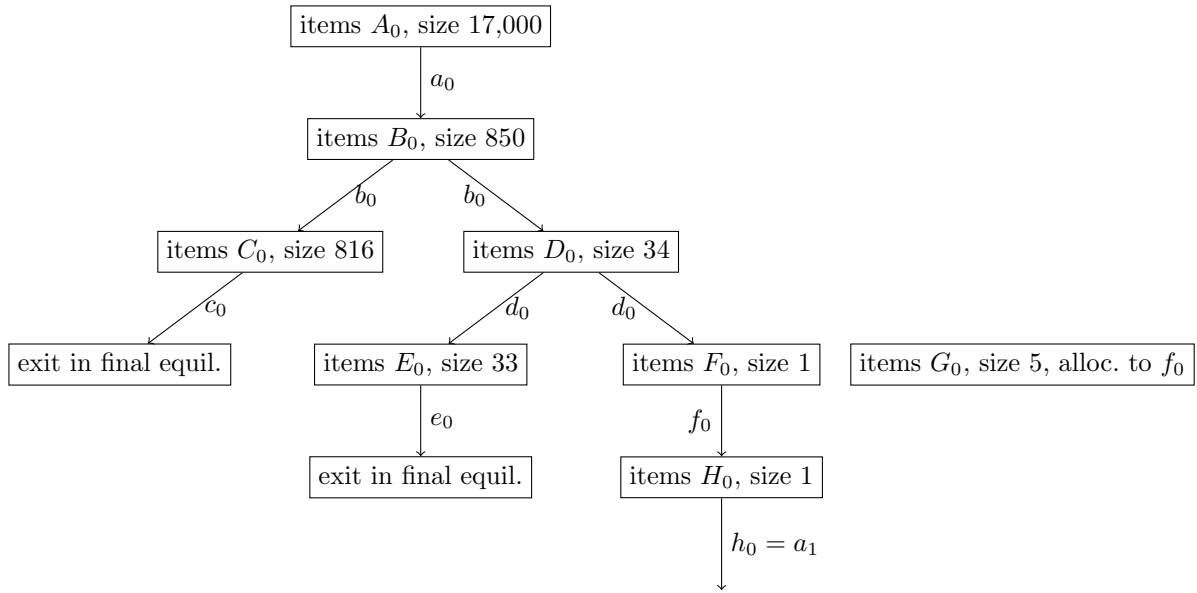

\begin{table*}[htbp]
\centering
\begin{tabular}{ccccccccccc}
item &    & $A_0$ & $B_0$ & $C_0$ & $D_0$ & $E_0$ & $F_0$ & $G_0$ & $H_0=A_1$ & $q$ value\\
size &    & 17,000 & 850 & 816 & 34 &  33 & 1 & 5 & 1 & \\ \hline\hline
\multicolumn{6}{l}{\bf The initial equilibrium:}&&&&\\
Bidder & size &&&&&&&& & 0\\
$a_0$	& 17,000 & {\bf 1} & 1.5 & 0 & 0 & 0 & 0 & 0 & 0 & 1\\
$b_0$   & 850    & 0 & {\bf 1}   & $\frac{4687}{7008}$ & 1.5 & 0 & 0 & 0 & 0 & 0.5\\
$c_0$   & 816 & 0 & 0 &  {\bf 1} & 0 & 0 & 0 & 0 & 0 & 0\\
$d_0$    & 34 & 0 & 0 & 0 &  {\bf 1} & $5/11$ & 2 & 0 & 0 & 0\\
$e_0$    & 33 & 0 & 0 & 0 & 0 & {\bf 1} & 0 & 0 & 0 & 0\\
$f_0$    & 6  & 0 & 0 & 0 & 0 & 0 & \boldmath{$\frac 83$} & \boldmath{$\frac 23$} & $\frac 53$ & $\frac{2}{3}$\\
$h_0$    & 1  & 0 & 0 & 0 & 0 & 0 &     0 & 0 &{\bf 1} & 0 \\ \hline
$t$ value &   & 0 & 0.5 & 1 & 1 & 1 & 2 & 0 & 1\\ 
\\
\multicolumn{6}{l}{\bf The final equilibrium:}&&&&\\
$a_0$	& 17,000 & \boldmath{$\frac{40}{41}$} & \boldmath{$\frac{60}{41}$} & 0 & 0 & 0 & 0 & 0 & 0 & $\frac{40}{41}$\\
$b_0$   & 850    & 0 & $\frac{292}{205}$   & \boldmath{$\frac{4687}{4920}$} & \boldmath{$\frac{438}{205}$}  & 0 & 0 & 0 & 0 & $\frac{192}{205}$\\
$d_0$    & 34 & 0 & 0 & 0 & 2 & \boldmath{$\frac{10}{11}$} & {\bf 4} & 0 & 0 & $\frac{4}{5}$\\
$f_0$   & 6  & 0 & 0 & 0 & 0 & 0 & $\frac{16}{5}$ & \boldmath{$\frac{4}{5}$} & {\bf 2} & 0\\
$h_0= a_1$    & 1  & 0 & 0 & 0 & 0 & 0 &     0 & 0 & 2 & 0 \\ \hline
$t$ value &   & 0 & $\frac{20}{41}$ & $\frac{79}{4920}$ & $\frac{6}{5}$ & $\frac{6}{55}$ & $\frac{16}{5}$ & $\frac{4}{5}$ & $2= \pAFa$\\ 
\end{tabular}
\caption{\label{tbl::Co}Market $M_0$, showing normalized valuations, multiplicity of bidders and items (their sizes), assignments (in bold), and the $t$ and $q$ values, for both the initial and final equilibria. The overlap with Market $M_1$ lies in item $H_0$ which is also item $A_1$ and bidder $h_0$ who is also bidder $a_1$.
Note that the $t$ values for $H_0=A_1$ are the same in markets $M_0$ and $M_1$ in both equilibria, as are the $q$ values for $h_0 = a_1$.}
\end{table*}

We continue by presenting the constructions of markets $M_r$, for $1 \le r < s$ and
of market $M_s$ in Tables~\ref{tbl::Cr} and~\ref{tbl::Cs}, respectively. 
$M_s$ is very similar to $M_r$; the only difference lies
in the presence of one additional item $I_s$, which is the item $e_s$, the losing
bidder,  will receive in the final equilibrium, plus one additional bidder, $i_s$,
who leaves in the final setting.
In these markets, all items are fully allocated in both equilibria.

To complete the construction we have to show that the various unspecified 
parameters
can be chosen so that the conditions of 
\eqref{tight-utility-rule} are satisfied for every item-bidder pair.
(It is immediate that~\eqref{p-q-constraints}--\eqref{q-positive}
are satisfied.)

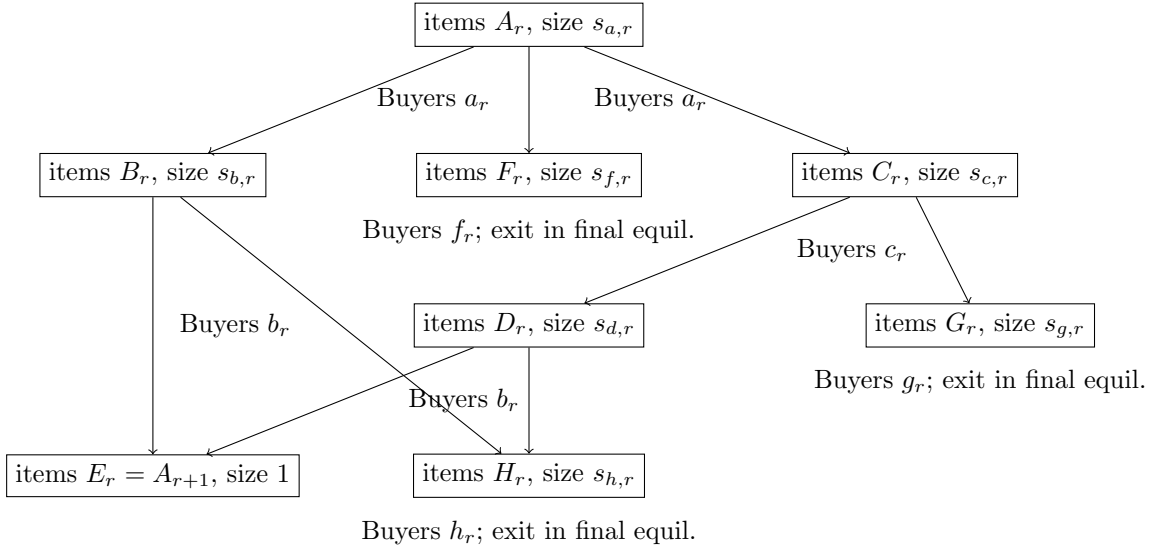
\begin{figure*}[htbp]
\begin{center}
\begin{tikzpicture}
\path (0,0) node [rectangle,draw](a) {items $A_r$, size $\sar$} --
(-5,-2) node[rectangle,draw](b) {items $B_r$, size $\sbr$} --
(0,-2) node[rectangle,draw](f) {items $F_r$, size $\sfr$} --
(-0, -2.75) node(ff) {Buyers $f_r$; exit in final equil.} --
(5,-2)  node[rectangle,draw](c) {items $C_r$, size $\scr$} --
(0,-4)  node[rectangle,draw](d) {items $D_r$, size $\sdr$} --
(6,-4)  node[rectangle,draw](g) {items $G_r$, size $\sgr$} --
(6, -4.75) node(gg) {Buyers $g_r$; exit in final equil.} --
(-5,-6)  node[rectangle,draw](e) {items $E_r= A_{r+1}$, size $1$}--
(0,-6)  node[rectangle,draw](h) {items $H_r$, size $\shr$}--
(0, -6.75) node(hh) {Buyers $h_r$; exit in final equil.};
 \draw[->] (a) -- node[right] {~~~Buyers $a_r$} (b);
  \draw[->] (a) --  (f);
 \draw[->] (a) -- node[left] {Buyers $a_r$~} (c);
 \draw[->] (b)  -- node[right] {~~Buyers $b_r$} (e);
 \draw[->] (b) --  (h);
 \draw[->] (c)  -- node[right] {~~~~~~~~Buyers $c_r$} (d);
  \draw[->] (c)  --  (g);
 \draw[->] (d)  --  (e);
  \draw[->] (d) --  node[left] {Buyers $b_r$}(h);
\end{tikzpicture}
\caption{\label{fig::Mr-flow}The Allocations in Market $M_r$, $1 \le r < s$.
The parameters are specified in the proof of Lemma~\ref{lem::param-choice}.}
\end{center}
\end{figure*}

\begin{table*}
\centering
\begin{tabular}{ccccccccccc}
item &    & $A_r$ & $B_r$ & $C_r$ & $D_r$ & $E_r$ & $F_r$  & $G_r$ & $H_r$ & $q$ value\\
       &    &          &          &          &          & $=A_{r+1}$ &  &  \\
size &    & $\sar$ & $\sbr$ & $\scr$ & $\sdr$ & $1$ & $\sfr$ & $\sgr$ & $\shr$ & \\ \hline\hline
\multicolumn{6}{l}{\bf Initial equilibrium:}&&&&\\
Bidder & size &&&&&& &&& \\
$a_r$	& $\sar $ & {\bf 1} & 2 & $\frac 12$ & 0 & 0 & $\vIfr$ & 0 & 0 &  0\\
$b_r$   & $\sbr+\sdr$    & 0 & \boldmath{$\frac {16}{7}$}   & 0 &  \boldmath{$\frac {2}{7}$}  & $ \frac {9}{7}$ & 0  & 0 &$\vIhr$ & $\frac{2}{7}$\\
$c_r$   & $\scr$ & 0 & 0 &\boldmath{$1$} & $\frac 12$ & 0 & 0 & $\vIgr$ & 0 & $\frac 12$\\
$f_r$ & $\sfr$ & 0 & 0 & 0 & 0 & 0 & {\bf 1} & 0 & 0 & 0\\
$g_r$ & $\sgr$ & 0 & 0 & 0 & 0 & 0 & 0 & {\bf 1} & 0 & 0\\
$h_r$ & $\shr$ & 0 & 0 & 0 & 0 & 0 & 0 & 0 & {\bf 1} & 0\\
\hline
$t$ value &   & 1 & 2 & $\frac 12$ & 0 & $1$ & 1 & 1 & 1\\ 
&&&&&&&&\\
\multicolumn{6}{l}{\bf Final equilibrium:}&&\\
$a_r$	& $\sar$ & \boldmath{$\frac{8}{9}\vr$} & \boldmath{$\frac{16} {9}\vr$} & \boldmath{$\frac{4} {9}\vr$}& 0 & 0 & \boldmath{$\vFfr$} &  0 & 0 &0 \\
$b_r$ & $\sbr+\sdr$ & 0 & $\frac{16} {9}\vr$ & 0 & $\frac{2} {9}\vr$  & \boldmath{$\vr$} & 0 & 0 & \boldmath{$\vFhr$} & 0
\\ 
$c_r$    & $\scr$  & 0 & 0 & $\frac{4} {9}\vr$ & \boldmath{$\frac{2} {9}\vr$} & 0 & 0  & \boldmath{$\vFgr$} & 0 & 0\\
\hline
$t$ value  &   & $\frac{8}{9}v_r$ & $\frac{16}{9}v_r$ & $\frac{4}{9}v_r$ & $\frac{2}{9}v_r$ & $v_r$  & $\vFfr$ & $\vFgr$ & $\vFhr$\\ 
\end{tabular}
\caption{\label{tbl::Cr}Component $C_r$, showing normalized valuations, multiplicity of bidders and items (their sizes), and the $t$ and $q$ values, for both the initial and final equilibria.
$\valFar$ and $\valFcr$ are normalizing factors, equal to the value of the
assignments using the initial valuations.}
\end{table*}

\begin{figure}[htbp]
\begin{center}

\begin{tikzpicture}
\path (0,0) node[rectangle,draw](e) {items $E_s$, size $1$}--
(0,-2) node[rectangle,draw](i) {items $I_s$, size $1$} --
(0, -2.75) node(ii) {Buyer $i_s$; exit in final equil.};
 \draw[->] (e) -- node[right] {Buyer $e_s$, the loser} (i);
\end{tikzpicture}
\caption{\label{fig::Ms-flow}The Allocations in Market $M_s$.}
\end{center}
\end{figure}
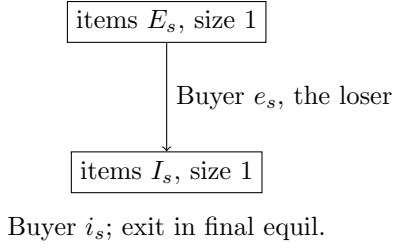

\begin{table}
\centering
\begin{tabular}{ccccc}
item &   & $E_s$ & $I_s$ & $q$ value\\
size &    & $1$ & $1$      &  \\ \hline\hline
\multicolumn{4}{l}{\bf The initial equilibrium:}\\
Bidder & size &&& \\
$e_s$& 1	&  {\bf 1} & $1/(v_s+1)$ & 0\\
$i_s$   & 1    & 0 & {\bf 1}   & 0\\ \hline
$t$ value &   & 1 & 1  &\\ 
&&&&\\
\multicolumn{4}{l}{\bf The final equilibrium:}\\
$e_s$	& 1   & $v_s+1$ & 1 & 1\\ \hline
$t$ value &   &  $v_s$   & 0  & \\ 
\end{tabular}
\caption{\label{tbl::Cs}Component $C_s$, showing the additional portion in addition to the part shown in Table~\ref{tbl::Cr}.}
\end{table}

\begin{lemma}
\label{lem::param-choice}
There are choices of values for the unspecified parameters for which
the valuations specified above yield
the claimed initial and final equilibria.
\end{lemma}
\begin{proof}
We need to choose the values $\vIfr$, $\vIgr$, $\vIhr$, $\vFfr$, $\vFgr$, 
$\vFhr$, the sizes $k_r$,
and the size proportionality factors $\sar$, $\sbr$, $\scr$, $\sdr$, $\sfr$, 
$\sgr$, $\shr$ so that
for each buyer in each equilibrium, its average value is 1, for $1\le r \le s$.
First, we set $\sar = \sbr + \sfr + \scr$,
$\scr = \sdr + \sgr$,
and $ \sbr + \sdr = 1 + \shr$.

Next, we observe that because the buyer $h_0=a_1$, its values
for $H_0=A_1$ are the same, i.e.\ $2 = \frac 89 v_1$, or $v_1 = 2 \cdot \frac 98$.
Similarly, item $E_r = A_{r+1}$, so $v_{r+1} = \frac 98 v_r$.
We conclude that $v_r = 2 \cdot \left(\frac 98 \right)^r$.

Now, for buyer $b_r$, we choose $\sbr = \frac 59 \sdr$, for 
\begin{align*}
\frac {\frac {16}{7}\sbr\kr + \frac 27 \sdr\kr} {\sbr\kr + \sdr\kr} 
=\frac{ \frac {16}{7} \cdot \frac 59 + \frac 27} {\frac 59 + 1} = 1.
\end{align*}
Thus $\sdr = \frac 9{14} (1 + \shr)$ and $\sbr = \frac 5{14} (1 + \shr)$.
To ensure these values are integers, we will make sure that $1 + \shr$ is an integer multiple of 14.

We turn to the values $\vFfr$, $\vFgr$, $\vFhr$.
We choose $\shr = 14 \floor{v_r/14}+ 13$, and $\vFhr$ to satisfy
$\vFhr \shr + v_r = \shr + 1$; i.e.\ 
$\vFhr = (\floor{v_r} +1 -v_r)/\floor{v_r}$.
We need to confirm that $\vIhr \le \frac 97$;
but $\vIhr = 9/(7v_r) \cdot \vFhr < \vFhr < 1$, as $v_r \ge \frac 94$.

Similarly, when $v_r > \frac 92$, we set $\sgr =  \floor{\frac29 v_r} \sdr$ 
(for when $v_r \le \frac 92$, this would set $\sgr = 0$), and
$\vFgr = \sgr  + \frac 29 v_r \sdr = \sgr + \sdr$;
i.e.\ $\vFgr = (\floor{\frac 29 v_r} \sdr +\sdr - \frac 29 v_r \sdr)/\floor{\frac 29 v_r} \sdr$.
Again, we need to confirm that $\vIgr \le \frac 32$;
but $\vIgr = 9/(4v_r)\cdot \vFgr \le \vFgr \le 1$, as $v_r > \frac 92$.

When $v_r < \frac 92$, we set 
$\sgr = \ceil{ \frac {9 - 2v_r} { \frac 23 v_r - 1} }$ and
$\vFgr = (\sgr + \sdr - \frac 29 v_r \sdr)/ \sgr$;
but then $\sdr = 9$, so 
$\vFgr = 1 + (9 - 2v_r)/\ceil{ \frac {9 - 2v_r} { \frac 23 v_r - 1} }$.
Again, we need to confirm that $\vIgr \le \frac 32$;
but $\vIgr = 9/(4v_r)\cdot \vFgr \le 9/(4v_r)(1 + \frac 23v_r -1) = \frac 32$.

As $v_r = 2 \cdot \left( \frac 98 \right)^r$, $v_r \ne \frac 92$ for any $r$.

Also, we set
$\sfr =  \floor{\frac {16}9 v_r}  \sbr + \floor{\frac 49 v_r} \scr$ and
$\vFfr \sfr  + \frac {16}9 v_r \sbr + \frac 49 v_r \scr = \sfr + \sbr + \scr$;
i.e.\ $\vFfr = (\floor{\frac {16}9 v_r}  \sbr + \floor{\frac 49 v_r} \scr + \sbr + \scr - ( \frac {16}9 v_r \sbr +  \frac 49 v_r \scr))/ (\floor{\frac {16}9 v_r} \sbr + \floor{\frac 49 v_r} \scr)$.

We can now calculate the following values.
\begin{align*}
\shr & =  14\floor{v_r/14} + 13\\
\sdr &=9(\floor{v_r/14} +1)\\
\sbr &= 5(\floor{v_r/14}+1)\\
\sgr &= 9\floor{2v_r/9}(\floor{v_r/14}+1)~~\text{for}~v_r > \frac 92\\
\sgr &= \ceil{ \frac {9 - 2v_r} { \frac 23 v_r - 1} }~~\text{for}~v_r < \frac 92\\
\scr &= 9(\floor{v_r/14}+1)\cdot(\floor{2v_r/9}+1)~~\text{for}~v_r > \frac 92\\
\scr &= 9 + \ceil{ \frac {9 - 2v_r} { \frac 23 v_r - 1} }~~\text{for}~v_r < \frac 92\\
\sfr &= 5\floor{16v_r/9}(\floor{v_r/14}+1) \\
&~~~+ 9\floor{4v_r/9}(\floor{v_r/14}+1)\cdot(\floor{2v_r/9}+1)~~\text{for}~v_r > \frac 92\\
\sfr &= 5\floor{16v_r/9}(\floor{v_r/14}+1)~~\text{for}~v_r < \frac 92\\
\sar & = 5(\floor{v_r/14}+1) + 9(\floor{v_r/14}+1)\cdot(\floor{2v_r/9}+1) \\
&~~~+ 5\floor{16v_r/9}(\floor{v_r/14}+1)\\
&~~~ + 9\floor{4v_r/9}(\floor{v_r/14}+1)\cdot(\floor{2v_r/9}+1)~~\text{for}~v_r > \frac 92\\
\sar & = 14 + \ceil{ \frac {9 - 2v_r} { \frac 23 v_r - 1} } 
+ 5\floor{16v_r/9}~~\text{for}~v_r < \frac 92.
\end{align*}
We also set $k_{r-1} = \sar k_r$ for $0 \le r < s$, and create $k_0$ copies of $M_0$. Recall that $k_s = 1$.
\end{proof}

To conclude the lower bound analysis we lower bound the size of $s$ and hence of $v_s$.
We observe that for $v_r > \frac 92$, $\sar \le  5(v_r/14 +1) +9(v_r/14 +1 )\cdot (2v_r/9 +1) + 5\cdot 16v_r/9 (v_r/14 +1) + 4v_r \cdot  (v_r/14 +1) \cdot (2v_r/9 +1) \le 4v_r^3/63 +91v_r^2/63 + 143v_r/9 +14 \le 2 v_r^3$, as $v_r \ge \frac 92$,
and for $v_r < \frac 92$, $\sar \le 14 + 9 + 80v_r/9 \le 4v_r^3$.

Note that $v_s = 2\left(\frac 98\right)^s$.
We can conclude that
\begin{align*}
k_0 & \le 4v_s^3 \cdot 4 \left (\frac 89 v_s\right)^3 \ldots
4 \left( \left(\frac 89\right)^{s-1}  v_s\right)^3 \\
&=4\left( 2 \left( \frac 98 \right)^s\right)^3
\cdot 4\left(2\left( \frac 98 \right)^{s-1}\right)^3 \ldots
\cdot \left( 2 \left( \frac 98 \right)\right)^3\\
&= (32)^s \left( \frac 98 \right)^{s(s-1)/2}\\
&\le \left( \frac 98 \right)^{(s^2+58s)/2}.
\end{align*}
Thus the size of item $A_0$ is at most $17000\left( \frac 98 \right)^{(s^2+58s)/2}$.
But this is more than $n/2$ by inspection of the construction.
Therefore  we can choose $s$ in our construction to satisfy 
$17000\left( \frac 98 \right)^{(s^2+58s)/2} \ge n/2$;
thus 
$$s^2 +58s)/2 \cdot \log \frac 98 + \log 17000 \ge \log n - 1.$$ and hence 
$$(s +29)^2 -29^2\ge 2(\log n -1- \log 17000)/\log \frac 98$$
So,  
$$s \ge (2[(\log n) -1 - \log 17000 +142] /\log \frac 98)^{1/2} - 29.$$

This implies 
\begin{align*} 
v_s &\ge \left(\frac 98\right)^{(2\log n/\log \frac 98)^{1/2} - 29}\\
&\ge 2^{\sqrt{2\log \frac 98 \log n} - 29 \log\frac 98}\\
&\ge 2^{ \sqrt{\log n}/2 - 5}.
\end{align*}

\end{document}